\documentclass{bioinfo}
\copyrightyear{2013}
\pubyear{2013}
\usepackage{multirow}
\usepackage{amsthm}
\usepackage{graphicx}
\usepackage{algorithm}
\usepackage{algorithmic}
\usepackage{sidecap}
\usepackage{wrapfig}
\usepackage{epsfig}
\usepackage{psfrag}
\usepackage{url}
\usepackage{amssymb}
\usepackage{amsmath}
\usepackage{setspace}
\newtheorem{lemma}{Lemma}

\usepackage{comment}  

\begin{document}
\firstpage{1}

\title[Cancer Diagnosis with QUIRE]{Cancer Diagnosis with QUIRE: QUadratic Interactions among infoRmative fEatures}
\author[Chowdhury \textit{et~al}]{Salim Akhter Chowdhury\,$^{1}$, Yanjun Qi\,$^{2}$, Alex Stewart\,$^{3}$, Rachel Ostroff\,$^{3}$ and Renqiang Min\,$^{2}$\footnote{to whom correspondence should be addressed}}

\address{$^{1}$Lane Center for Computational Biology, Carnegie Mellon University, Pittsburgh, PA, 15213, USA.\\
$^{2}$Department of Machine Learning, NEC Laboratories America, Princeton, NJ, 08540, USA. \\
$^{3}$SomaLogic, Inc., 2945 Wilderness Pl., Boulder, CO 80301, USA
}

\history{Received on XXXXX; revised on XXXXX; accepted on XXXXX}

\editor{Associate Editor: XXXXXXX}

\maketitle

\begin{abstract}
\noindent\textbf{Motivation:} Responsible for many complex human diseases including cancers, disrupted or abnormal gene interactions can be identified through their expression changes correlating with the progression of a disease. However, the examination of all possible combinatorial interactions between gene features in a genome-wide case-control study is computationally infeasible as the search space is exponential in nature.

\noindent\textbf{Results:} In this paper, we propose a novel computational approach, QUIRE, to identify discriminative complex interactions among informative gene features for cancer diagnosis. QUIRE works in two stages, where it first identifies functionally relevant feature groups for the disease and, then explores the search space capturing the combinatorial relationships among the genes from the selected informative groups. Using QUIRE, we explore the differential patterns and the interactions among
informative gene features in three different types of cancers, Renal Cell Carcinoma(RCC), Ovarian Cancer(OVC) and Colorectal Cancer (CRC). Our experimental results show that QUIRE identifies gene-gene interactions that can better identify the different cancer stages of samples and can predict CRC recurrence and death from CRC more successfully, as compared to other state-of-the-art feature selection methods.
A literature survey shows that many of the interactions identified by QUIRE
play important roles in the development of cancer.
\end{abstract}
\vspace{-0.2 in}
\section{Introduction}
\vspace{-0.01 in}

In this paper, we focus on the task of cancer diagnosis using molecular signatures,  
such as gene expression measured using microarray experiments or protein expression  
values measured in blood.  
In the past decade, advancement in the genome-wide monitoring of
gene expression has enabled the scientists to look into the interplay among multiple genes
and their products. Differential analysis of gene expression
helps identification of individual genes that show altered behavior
in the phenotype of interest. Previous systematic analysis of gene expression
has enabled identification of single gene markers implicated in different
types of cancers such as breast cancer \citep{perou2000},  
lung cancer \citep{beer2002}, and prostate cancer \citep{lapointe2004}.

Complex diseases like cancers are the results of multiple genetic and epigenetic
factors. Although single gene markers can provide valuable information
about the process under study,
a major problem with these markers is that they
offer limited insight into the complex
interplay among the molecular factors responsible for progression
of complicated diseases, like cancers. So, recently, research
focus in complex diseases shifts towards the identification of multiple
genes that interact directly or indirectly in contributing their association
to the target disease. Several complex interactive partners from previous studies
have been shown to give
important insight into the mechanism of breast cancer \citep{Chuang2007}
and colorectal cancer \citep{Chowdhury2011}.

However, due to the combinatorial nature of the problem,
the identification of groups of genes that show differential
behavior in the manifestation of complex phenotypes is computationally
infeasible. For instance, for a set of $30,000$ genes, there are about $4500$ million possible gene-gene interactions in the search space.  
Several recent methods propose to reduce the search space using orthogonal prior knowledge
about connections amongst the genes, such as interactions collected from protein-protein
interaction (PPI) network \cite{Lee and Xing, Bioinformatics, Vol. 28 ISMB 2012} or grouping information from functional annotations of proteins.  
One notable computational method named Group Lasso, which is proposed by \citealp{Yuan2006}, incorporates such prior interaction or grouping among the genes  
to detect gene groups that contribute to human disease,   
by enforcing sparsity at the group level in a supervised regression framework.
Group Lasso is extended
by \citealp{Jacob2009} to a more general setting that incorporates groups whose
overlaps are nonempty. Such overlaps in groups is biologically
more significant, because many genes participate in multiple pathways
and manifest themselves in several biological processes.

Although Group Lasso is very useful in identifying biologically
relevant groups of genes and proteins, they cannot capture complex
combinatorial relationships among the features within and across the groups.
Also, current PPI network data is inherently noisy due to experimental constraints
\citep{yu2008high}. Algorithmic approaches based solely on the noisy prior information can
result in many false positive interactions which are absent in the real genome space.

In this paper, our main goal is to identify the complex combinations of pairwise interactions among the genes that might help us (1) better diagnosis and prognosis of different types of cancer, and (2) gain novel insights into the mechanistic basis of the diseases. Since the total number of possible pairwise human gene interactions is huge, it is computationally infeasible to
examine all possible combinatorial combinations of them when trying to understand their relevance to the phenotype under consideration. Due to this ``Curse of Dimensionality'' issue, our first target is to reduce the dimensionality of the search space in such a manner that it enables us to identify informative interacting gene partners in a reasonable
limit of time and memory space. This reduced search space then enables us to look for combinations of interacting
pairs of informative genes in a more practical sparse learning setting.

In this paper, we propose a two-stage solution, named as QUIRE, i.e. to detect QUadratic Interactions among infoRmative fEatures
which show differential behavior for diagnosing a target disease using molecular signatures.
In the first stage of our proposed approach, we use Overlapping
Group Lasso \citep{Jacob2009} to identify biologically relevant informative feature groups and physical gene interaction groups that exhibit differential patterns for the studied disease. Then
in the second stage, we search exhaustively on this reduced feature space
by examining all possible pairs of interacting features to identify the combination of  
markers and complex patterns of feature interactions that are informative about the phenotypes in a sparse learning framework.
Systematic experimental results have been obtained on protein and gene expression data from three different types of cancers, Renal Cell Carcinoma (RCC), Ovarian Cancer (OVC) and
Colorectal Cancer (CRC). Compared with other state-of-the-art feature selection methods, the results show that QUIRE can discover complementary sets of markers and pairwise interactions that can better classify samples from different stages of cancer and predict post-cancerous events, like cancer
recurrence and death from cancer with higher accuracy.

Broadly speaking, the proposed method also connects to statistical approaches that identify
gene-gene interactions on genome-wide association studies, reviewed in
\citep{cordell2009detecting}. For example, Wu et al. have proposed \citep{Wu2009} a Lasso-penalized logistic regression to identify
pairwise interactions from genome-wide association data sets.
The limitation of their approach is that it only works on SNP data and thus
restricts the predictors only to three relevant values -1, 0 and 1. As a result, their
method is not computationally feasible and not directly applicable to the real-valued
gene expression data. To the best of our knowledge, QUIRE is the first proposed method to identify combinatorial patterns among the pairs of informative genes for studying complex diseases, like cancer. Subsequent functional analysis of the  
interactions identified by QUIRE reveals that it can indeed identify genes relevant to the progress of diseases under study.

\vspace{-.3in}
\begin{methods}
\section{Methods}

Mathematically, the identification of single gene markers in a genome-wide study is an ill-posed problem.  
This is due to the fact that the number of genes in human cells are much more in numbers than the number  
of samples that are available for these kind of studies. For such
problems, Lasso, proposed by \citealp{Tibshirani1996} is very
popular for selecting a  small number of features relevant to
the problem under study. When a set of features are highly correlated
to each other, Lasso selects one from that set randomly, ignoring
others. So, in our current setting, there is a possibility that Lasso leaves out biologically
relevant genes from its set of selected informative features.

In this paper, we consider a linear regression setting. Suppose we have a data set $D$ containing $n$ observations $({\mathbf x}^{(i)}, y^{(i)})$ with response variable $y \in R$ and feature vector ${\mathbf x} \in R^p$, where $i \in \{1, \ldots, n\}$,  and we assume that features are standardized with zero mean and unit standard deviation (we will prove the effects of feature standardization on linear classifiers later in this section) and the $y$s are centered in $D$. The Lasso approach optimizes the following objective function,
\begin{eqnarray}
\ell({\mathbf w}) & = & \sum \limits_{i=1}^n (y_i - \sum \limits_{j=1}^p w_j x^{i}_{j})^2, \nonumber \\
\ell_{lasso}({\mathbf w})  & = & \ell({\mathbf w}) + \lambda \sum \limits_{j=1}^p \lvert w_j \rvert,
\end{eqnarray}
where $\ell({\mathbf w})$ is the loss function of linear regression, and ${\mathbf w}$ is the weight parameter. The $\ell_1$ norm penalty in lasso induces sparsity in the weight space for selecting features. It is obvious that the sum of the least squared errors and the $\ell_1$ norm are convex functions with respect to the weights ${\mathbf w}$, thus, we have the following lemma,
\begin{lemma}
Lasso-penalized linear regression has global optimum for any fixed penalty coefficient $\lambda$.
\end{lemma}
According to Lemma 1, Lasso has global optimum, which can be found by any convex optimization technique. The coordinate descent approach proposed in \citealp{Friedman2009,VanderKooij2007} sets the gradient of the loss function $l_{lasso}({\mathbf w})$ to 0 to solve each weight $w_j$ iteratively, and it is among one of the most computationally efficient methods.
\begin{equation}
w_j = S(\frac{1}{n} \sum \limits_{i=1}^n x^{(i)}_{j} (y^{(i)} - \sum_{k \neq j}w_k x^{(i)}_j), \lambda)_+,
\end{equation}
where $S(z, \lambda)_+$ is a soft-thresholding operator. The value of $S(z, \lambda)_+$  is $z - \lambda$ if  $z>0$ and $\lambda < \lvert z \rvert$, $z + \lambda$ if  $z<0$ and $\lambda < \lvert z \rvert$, and 0 if $\lambda \geq \lvert z \rvert$.

In spite of the computational efficiency and the popularity of Lasso for feature selection, its formulation prevents it from capturing any prior information on possible group structures among the features. Group Lasso \citep{Yuan2006} proposed using $\ell_{2,1}$ penalty to select groups of input features which are partitioned into non-overlapping groups. The group penalty is the sum of the $\ell_2$ norm on the features belonging to the same group.
Overlapping Group Lasso \citep{Jacob2009} extends Group Lasso to handle groups of features with overlapping group members by duplicating input features belonging to multiple groups in the design matrix. Because a lot of real applications involve overlapping feature groupings, Overlapping Group Lasso is a more natural choice than Group Lasso. Suppose that we partition $p$ features in data set $D$ into $q$ overlapping groups $G = \{g_1, g_2, \ldots, g_q\}$, the following objective function is minimized in \citealp{Jacob2009} and \citealp{friedman2010note},
\begin{equation}\label{oglasso}
\ell_{oglasso} = \ell({\mathbf w}) + \lambda \sum_{g \in G}|| {\mathbf w}_g ||_2,
\end{equation}
where $\lambda$ is the regularization parameter, ${\mathbf w}_g$ denotes the set of weights associated with features in group $g$, and $||\cdot||_2$ is the Euclidean norm. The above optimization problem is separable, so we can use block coordinate descent to optimize the weights associated with each group $g$ separately. The subgradient of the optimization takes the following form,
\begin{equation}
- \sum_{i=1}^n {x}_g^{(i)T} (y^{(i)} - \sum_{g^\prime} {\mathbf w}_{g^\prime} x^{(i)}_{g^\prime}) + \lambda \frac{{\mathbf w}_g}{||{\mathbf w}_g||} = 0; \forall g \in G.
\end{equation}
Therefore, if $||\sum_{i=1}^n {x}_g^{(i)T} (y^{(i)} - \sum_{g^\prime \neq g} {\mathbf w}_{g^\prime} x^{(i)}_{g^\prime})|| < \lambda$, then ${\mathbf w}_g$ = 0; otherwise, ${\mathbf w}_g$ can be obtained by solving several one-dimensional optimization problems based on coordinate descent. In details,  let $Z^{(i)} = x^{(i)}_g$ = $(Z^{(i)}_1, \ldots, Z^{(i)}_k)$, ${\mathbf w}_g$ = ${\mathbf \theta}$ = $(\theta_1, \ldots, \theta_k)$, and residual  $r^{(i)} = y^{(i)} -  \sum_{g^\prime \neq g} {\mathbf w}_{g^\prime} x_{g^\prime}^{(i)}$, then $\theta_j$s of ${\mathbf w}_g$ can be solved by minimizing the following objective function,
\begin{equation}
\frac{1}{2}\sum_{i=1}^n (r^{(i)} - \sum_{j=1}^k Z^{(i)}_j \theta_j)^2 + \lambda ||{\mathbf \theta}||_2.
\end{equation}
The final solution of the overlapping group lasso is obtained by iterating the above optimization procedure over each feature group $g$ until convergence.

\begin{figure}[htp]
\centering
\includegraphics[height=2.8in,width = 1.8in]{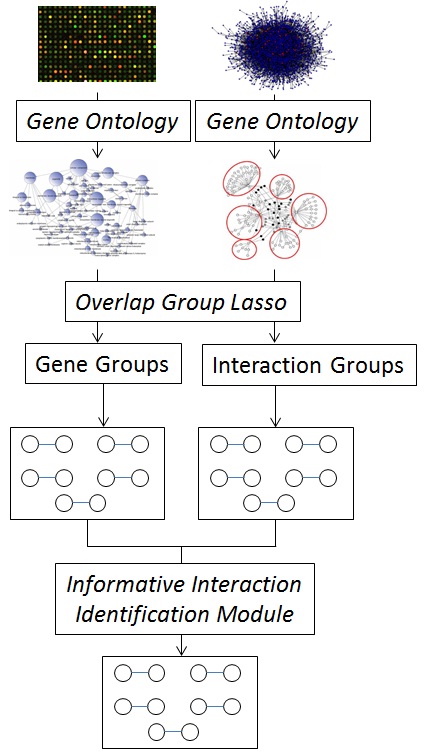}
\caption{Working model of QUIRE. QUIRE takes as input, gene or protein
expression levels of a set of samples, disease status of those samples and
physical interactions amongst the gene products. Then it uses gene ontology based
functional annotation to group the genes and cluster the interaction
network. Overlapping group lasso is run next on the expression and
interaction space to identify informative set of genes and interactions.
QUIRE then enumerates all pairwise binary interactions amongst the selected
gene features. Finally the proposed novel objective function is applied on the selected
single gene features, the informative protein protein interactions and the quadratic interactions amongst these genes to identify
the final set of interactions and gene markers.
}
\label{fig:met}
\vspace{-0.20 in}
\end{figure}

Although considering grouping structure among input features is very important for feature selection, Overlapping Group Lasso only encourages sparsity at the feature group level and there is no sparsity penalty within feature groups. Therefore, Overlapping Group Lasso often outputs a much larger number of selected features than Lasso.  Furthermore, Lasso and Overlapping Group Lasso only consider single gene features for prediction, which is very limited for disease status prediction and biomarker discovery.

For cancer diagnosis and biomarker discovery from blood samples or tissue samples, we plan to consider all possible combinations of single gene features and quadratic gene interaction features. Ideally, we want to optimize the following optimization problem to identify discriminative features given the dataset $D$,
\begin{eqnarray}
\ell({\mathbf w}, {\mathbf U}) & = & \sum \limits_{i=1}^n (y^{(i)} - \sum \limits_{j=1}^p w_j x^{(i)}_{j} - \sum \limits_{j=1}^{p-1} \sum \limits_{k=j+1}^p U_{jk} x^{i}_{j} x^{i}_{k} \nonumber \\
& & + \lambda_1 \sum \limits_{j=1}^m \lvert w_j \rvert + \lambda_2 \sum \limits_{j=1}^{p-1} \sum \limits_{k=j+1}^p \lvert U_{jk} \rvert.
\end{eqnarray}
However, the above model has $O(p^2)$ features and is not applicable to genome-wide biomarker discovery studies. Provided that the training data is often very limited, it is almost impossible to identify the discriminative single or quadratic interaction features by solving the above optimization problem. We propose QUIRE (QUadratic Interactions among infoRmative fEatures) to address these challenges, which is based on Overlapping Group Lasso and Lasso. And it takes advantage of both of these feature selection methods.

The underlying idea of QUIRE is to incorporate all possible complementary biological knowledge into the above infeasible optimization problem to reduce search space. By restricting discriminative gene interactions to happen only between genes in some informative gene groups, we can use existing functional annotations of input genes to identify these groups thereby to throw away a lot of interaction terms during the optimization. In addition, available physical interactions between the protein products of input genes can also be used to cut the search space, although discriminative gene feature interactions for prediction do not always necessarily correspond to physical interactions. The general working model of QUIRE is shown in Figure~\ref{fig:met}.
In details, QUIRE takes the expression profile of $n$ samples over $p$ genes (proteins),
the physical interactions among the genes products (i.e. protein protein interaction network)
and the disease status of these samples as input,  and it outputs a (small) set of discriminative genes and gene interactions with corresponding learned weights for predicting the disease status of any incoming test sample.  
The step by step working model of QUIRE is given below:

\begin{enumerate}

\item \textit{Functional group generation:}

\begin{enumerate}

\item QUIRE groups the $p$ input gene features into $q$ overlapping functional categories according to the existing Gene Ontology (GO) based functional annotations, such as Cellular Colocalization (CC), Molecular Function (MF), and Biological Process (BP).

\item QUIRE clusters the given interaction network (i.e. PPI) into subsets of overlapping gene
products based on GO functional annotations, CC, MF and BP.

\end{enumerate}

\item \textit{Informative genes and functional interactions selection:}
\begin{enumerate}
\item Given the GO functional grouping of input gene features, Overlapping Group Lasso is run to select $m$ top discriminative genes for disease status prediction according to the absolute values of the learned weights of gene features.
\item Overlapping group lasso is run on the clustered interaction network to select
informative groups of protein-protein interactions. In this case, each cluster
is considered as a group
and quadratic interactions (discussed later) among the interacting proteins in a
group are used as expression.
\end{enumerate}

\item \textit{Selection of most informative interactions and genes:} QUIRE first enumerates all possible quadratic feature interactions among the informative genes
selected at step 2(a). Then it takes these quadratic interactions, single informative gene features and
the informative functional interactions identified at step 2(b) as input and it outputs the final selected gene interactions and single genes as biomarkers.
\end{enumerate}

In order to identify the discriminative combinations of single gene features and quadratic interactions among pairwise informative genes, we define our proposed objective function for Lasso as follows,
\begin{eqnarray}
\ell({\mathbf w}, {\mathbf U},{\mathbf R}) & = & \sum \limits_{i=1}^n (y^{(i)} - \sum \limits_{j=1}^m w_j x^{(i)}_{j} - \sum \limits_{j=1}^{m-1} \sum \limits_{k=j+1}^m U_{jk} x^{i}_{j} x^{i}_{k} - \sum \limits_{l=1}^r R_l I_{l})^2 \nonumber \\
& & + \lambda_1 \sum \limits_{j=1}^m \lvert w_j \rvert + \lambda_2 \sum \limits_{j=1}^{m-1} \sum \limits_{k=j+1}^m \lvert U_{jk} \rvert + \lambda_3 \sum \limits_{l=1}^r \lvert R_l \rvert,
\end{eqnarray}

where $j$ and $k$ index the seed informative genes and $l$ indexes the informative protein protein interactions selected by the Overlapping Group Lasso in the previous step. The objective function contains $\ell_1$ penalties at single informative gene level, and pairwise gene interaction and protein interaction level. The intuition behind this formulation is that it captures the interactions that are complementary to the individual informative genes. Because it is computationally infeasible to consider every pair of interaction in a genome wide case control study, QUIRE reduces the search space by using the features that are selected by Overlapping Group Lasso as the informative ones, and then it relies on Lasso with $\ell_1$ penalties to identify the discriminative combination of informative individual gene features and gene interaction features, which provides an approximation to the problem of searching an exponential number ($O(2^{p+p^2})$) of all possible combinations of single features and pairwise interaction features.  

As we mention earlier in this paper, we perform feature standardization before running Lasso or Group Lasso. Instead of using the original quadratic interactions $x_j x_k$ between pairwise variables $x_j$ and $x_k$, we standardize $x_j x_k$ by $g(x_j x_k)$ as input feature, where $g(x) = \frac{ x - \mu }{ \sigma}$, and $\mu$ and $\sigma$ are respectively the mean and standard deviation of feature $x$. As shown in Lemma 2 and Lemma 3, we see that feature standardization has nice properties when running Lasso, and quadratic feature interactions calculated by $g(x_j x_k)$ is more sensible than $g(x_j) g(x_k)$ for biomarker discovery because it does not have weight sharing constraints involving both gene interaction features and  single gene features. Moreover, $g(x_j) g(x_k)$ can result in inaccurate calculations because the product of two large negative values for normalized features is a large positive value, which is not desirable in most applications. The advantage of $g(x_j x_k)$ over $g(x_j) g(x_k)$ is supported by our experimental results in this paper.

\begin{lemma}
The solution of Lasso-penalized linear regression on standardized input features with one fixed penalty coefficient $\lambda$ is equivalent to the solution of a Lasso problem on original input features with adaptive penalty coefficients for different weights being $\lambda$ weighted by the standard deviations of different corresponding original features.
\end{lemma}
\begin{proof}
Suppose that the optimal solution of lasso-penalized linear regression on standardized features takes the following linear form, $f(x; {\mathbf w}) = \sum_j w_j^T g(x_j) = \sum_j \frac{w_j}{\sigma_j }x_j - \sum_j w_j \frac{\mu_j}{\sigma_j}$, in which ${\mathbf w}$ minimizes $\ell({\mathbf w}) = \sum_i ||y^{(i)} - f(g(x^{(i)}); w)||^2 + \lambda \sum_j |w_j|$. Let $w^{\prime}_j = \frac{w_j}{\sigma_j}$, we can show that $w^{\prime}$ is the global optimal of $\ell^{\prime}({\mathbf w}^{\prime}) = \sum_i ||y^{(i)} - \sum_j w^{\prime}_j x^{(i)}_j ||^2 + \sum_j \lambda \sigma_j |w^{\prime}_j|$. Because $x_j = \sigma_i g(x_j) + \mu_j$, $\ell(w^{\prime}) = \sum_i || y^{(i)} - \sum_j \frac{w_j}{\sigma_j} (\sigma_i g(x_j) + \mu_j) ||$ + $\sum_j \lambda \sigma_j \frac{w_j}{\sigma_j}$ = $\sum_i ||y^{(i)} - f(g(x^{(i)}); w)||^2 + \lambda \sum_j |w_j|$ = $ \ell ({\mathbf w})$. Therefore, if ${\mathbf w}$ is a global optimum of $\ell({\mathbf w})$, then ${\mathbf w}^{\prime}$ must be a global optimum of $\ell^{\prime}({\mathbf w}^{\prime})$, which proves the lemma.
\end{proof}

\begin{lemma}
In the setting of Lasso-penalized linear regression, our proposed quadratic feature interaction $g(x_j x_k)$ has different effect compared to $g(x_j) g(x_k)$. $g(x_j x_k)$ only constrains original feature interactions $x_j x_k$ while $g(x_j) g(x_k)$ results in weight sharing constraints involving both interaction features and single features.
\end{lemma}
\begin{proof}
First, if $x_j x_k$ correlates with a feature $x_l$ , $g(x_j x_k)$ does not change the correlation coefficient, because $corr(x_j x_k, x_l)$ = $\frac{E[x_j x_k - \mu_{jk}] [x_l - \mu_l ]} {\sigma_{jk} \sigma_l}$ = $E[g(x_j x_k) g(x_l)]$ = $corr(g(x_j x_k), x_l)$. However, $corr(g(x_j) g(x_k) , x_l)$ is different from $corr(x_j x_k, x_l)$. Second, $g(x_j)g(x_k) = \frac{x_i x_j - \mu_i x_i - \mu_j x_j + \mu_i \mu_j}{\sigma_i \sigma_j}$ includes both a quadratic interaction term and linear terms, but $g(x_j x_k) = \frac{x_i x_j - \mu_{ij}}{\sigma_{ij}}$ only involves a quadratic interaction term. Therefore, linear classifier $\sum_{jk} w_{jk} g(x_j)g(x_k) + b$ puts a lot of constraints on weight sharing between quadratic interaction features and single features. Moreover, according to Lemma 2, the solution to Lasso on features $g(x_j x_k)$s is equivalent to the solution to the lasso on features $x_j x_k$s with adaptive penalty coefficients for different weights.
\end{proof}
\end{methods}

\vspace{-0.03 in}
\section{Results and Discussion}
\vspace{-0.01 in}

In this section, we perform comprehensive classification experiments to
evaluate the performance of our proposed method, QUIRE, in classifying
samples across different stages of Renal Cell Carcinoma (RCC), Ovarian Cancer
(OVC) and in predicting cancer recurrence and death due to cancer in Colorectal
Cancer (CRC). We compare the performance of QUIRE with state-of-the-art feature
selection techniques, Lasso, Overlapping Group Lasso and SVM. We then perform a
literature survey and enrichment analysis of the informative interactions selected by QUIRE and show that they are
relevant to the progression of the disease.

\subsection{Datasets}
\vspace{-0.02 in}

We use three datasets with samples from RCC, OVC and CRC for classification experiments. For RCC and OVC datasets, Blood samples are collected and Somamer (aptamer) based proteomic technology \citep{Gold2010} is used to measure the concentration of a selected set of marker proteins. The CRC samples belong to a publicly available microarray dataset collected from gene expression omnibus (GEO), and referenced by accession number $GSE17536$ \citep{Smith2010}. We give detailed description of the datasets below:  

\begin{enumerate}

\item RCC Dataset: This dataset contains $212$ RCC samples from Benign and $4$ different stages of tumor. Expression levels of $1092$ proteins are collected in this dataset. The number of Benign, Stage $1$, Stage $2$, Stage $3$ and Stage $4$ tumor
samples are $40, 101, 17, 24$ and $31$ respectively.

\item OVC Dataset: This dataset contains $845$ proteins' expressions for
$248$ samples across Benign and $3$ different stages of ovarian cancer. The number of Benign, Stage $1$, Stage $2$ and Stage $3$ tumor
samples are $134, 45, 8$ and $61$ respectively.

\item CRC dataset ($GSE17536$): This microarray dataset contains $177$ samples from $4$ different
stages (Stage $1$ to Stage $4$) of CRC. Expression levels of $20125$ genes are collected. Besides stage information,
this dataset also has records for each patient, the binary valued information of ``Cancer
Recurrence'' and ``Death from Cancer''. Out of $177$ patients, $55$ had recurrence of cancer and
$68$ died from cancer.

\end{enumerate}

\subsection{Grouping of Features}
\vspace{-0.02 in}
In order to group the genes using gene ontology terms, we use the web based tool
``Database for Annotation, Visualization, and Integrated Discovery'' (DAVID,
\texttt{http://david.abcc.ncifcrf.\\gov/})\citep{dennis2003david}. There are a set of parameters that can be
adjusted in DAVID based on which the functional classification is done.
This whole set of parameters is controlled by a higher level parameter
``Classification Stringency'', which determines how tight the resulting
groups are in terms of association of the genes in each group.
In general,  a ``High'' stringency setting generates less number of
functional groups with the member genes tightly associated and more
genes will be treated as irrelevant ones into an unclustered group.
We set the stringency level to ``Medium"
which results in balanced functional groups where the association of
the genes are moderately tight. For the CRC dataset, functional classification
using BP, CC and MF results in total $2434$, $520$ and $1146$ groups respectively.
The total number of groups for BP, CC and MF annotations on RCC and OVC
datasets are $890,56,155$ and $321,23,56$ respectively.

\subsection{Protein Protein Interactions}
\vspace{-0.02 in}
Besides using it for selecting informative single gene features,
we use Overlapping Group Lasso to select the informative protein
protein interactions.  
We download the binary protein protein interactions (PPI) data
from HPRD(\texttt{http://www.hprd.org/}). For each group $G_i$ in a particular functional
grouping of the genes (i.e. BP, CC or MF), we identify
the pairs of member genes of $G_i$ whose products interact directly with each other in the PPI network.
The set of all such pairs where both interacting partners are members of $G_i$ forms a group.    
For a pair of interacting genes
$x_j$ and $x_k$ in a group, we use their quadratic interaction term $g(x_j,x_k)$
as their expression level. Usage of the quadratic interaction formulation in Overlapping Group Lasso helps us to integrate the resulting informative protein protein interactions into the formulation of QUIRE directly without any transformation. Thus the total number of groups are same in the
case of interactions and single gene features. But the cardinality of each group and
the expression levels of the members are different.

\subsection{Experimental Design}
\vspace{-0.02 in}

In order to systematically evaluate the classification performance of QUIRE, we perform the following classification experiments:

\begin{enumerate}

\item Classification experiments using RCC samples: We perform three stage-wise binary classification experiments using RCC samples:

\begin{enumerate}

\item Classification of Benign samples from Stage $1-4$ samples.
\item Classification of Benign and Stage $1$ samples from Stage $2-4$ samples.
\item Classification of Benign, Stage $1,2$ samples from Stage $3,4$ samples.

\end{enumerate}

\item Classification experiments using OVC samples with intermediate levels of
\textit{CA125} as test set: \textit{CA125} is a well-known marker in ovarian cancer \citep{Suh2010}. Concentration of \textit{CA125} is used to measure the progression of the disease.
The suspicious cutoff level of \textit{CA125} is $40$, meaning that concentration level above $40$ of this marker might be indicative of OVC. But \textit{CA125} is not a good indicator of early detection of the disease onset, especially when the concentration of this biomarker is between $40$ and $100$.
So we use samples with \textit{CA125} concentration level between $40$ and $100$ as our test
set in this experiment. The remaining samples, with concentration of \textit{CA125} below $40$ and above $100$ are used as training set. We perform the following experiments:

\begin{enumerate}

\item Classification of Benign samples from Stage $1-3$ samples.
\item Classification of Benign, Stage $1$ samples from Stage $2, 3$ samples.
\item Classification of Benign, Stage $1, 2$ samples from Stage $3$ samples.

\end{enumerate}

\item Classification experiments using CRC samples: We perform two classification experiments using the samples with CRC:

\begin{enumerate}

\item Prediction of cancer recurrence: we build binary classifier to predict whether there is disease free survival in the follow-up time or not.

\item Prediction of death from colorectal cancer: we train binary classifier to predict if there is death from CRC across all pathological stages of the disease.

\end{enumerate}

\end{enumerate}

\subsection{Classification performance of QUIRE}
\vspace{-0.02 in}

In this section, we report systematic experimental results on classifying samples
from different stages of RCC and OVC and in predicting CRC recurrence and death
from CRC. In the first stage of QUIRE, we use Overlapping Group Lasso to
identify the biologically relevant groups of features and pairwise protein interactions, which in turn, is used in the subsequent stage to
explore the set of informative markers and quadratic interactions.
However, for the RCC and OVC datasets, we do not use protein protein
interactions for prediction purpose. This is because,
these datasets include only selected marker proteins distributed sparsely
across the protein interaction network and thus most of them do not interact
with each other directly.

After we run Overlapping Group Lasso on the gene groups, we sort the genes
based on the weight value assigned to it by the algorithm. We select
top $m = 200$ genes as input to QUIRE for enumerating all pairwise quadratic
interactions. We report the effect of $m$ on the classification performance of QUIRE later
in this section.
For classification of CRC samples, Overlapping Group Lasso on average selects
$1000$ protein protein interactions
as informative ones. We use this whole set of selected protein interactions
as input to QUIRE to be considered besides the paired quadratic interactions.
In this section, we present Overlapping
Group Lasso's and QUIRE's performance for features grouped
using Cellular Colocalization (CC) ontology, as this grouping strategy
shows best classification performance among the three we consider.

The predictive performance of the markers and pairwise interactions selected
by QUIRE is compared against the markers selected by Lasso, SVM and
Overlapping Group Lasso. We use glmnet \citep{Friedman2009} and LiblineaR \citep{Fan2008}
packages for implementation of Lasso and SVM respectively. We use the Group Lasso
implementation (with overlapping groups) from \citep{Jacob2009}. The overall performance
of the algorithms are shown in Figure~\ref{fig:res}. In this figure, we report average AUC score for ten
runs of five fold cross validation experiments for cancer stage prediction
in RCC (Figure~\ref{fig:res}(A)) and for predicting cancer recurrence and death from cancer in CRC(Figure~\ref{fig:res}(C)).
In five fold cross validation experiments, we divide the samples equally into five disjoint sets or folds.
We keep one fold for testing. On the remaining four folds, we use Overlapping Group Lasso to
identify the informative set of markers and protein protein interactions (for CRC). We train QUIRE on these four folds using these markers to identify the pairwise
interactions and markers and use the set-aside test set for prediction purpose. For each run, this procedure is repeated for each of the five folds and
average AUC score is reported for ten such runs. For OVC,
we report average AUC score (Figure~\ref{fig:res}(B)) for predicting the cancer stage of the samples
with intermediate levels of \textit{CA125} (concentration of \textit{CA125} is between $40$ and $100$) using the
remaining samples for training and informative feature selection.  

\begin{figure*}[htp]
\centering
\includegraphics[width = 0.9\textwidth]{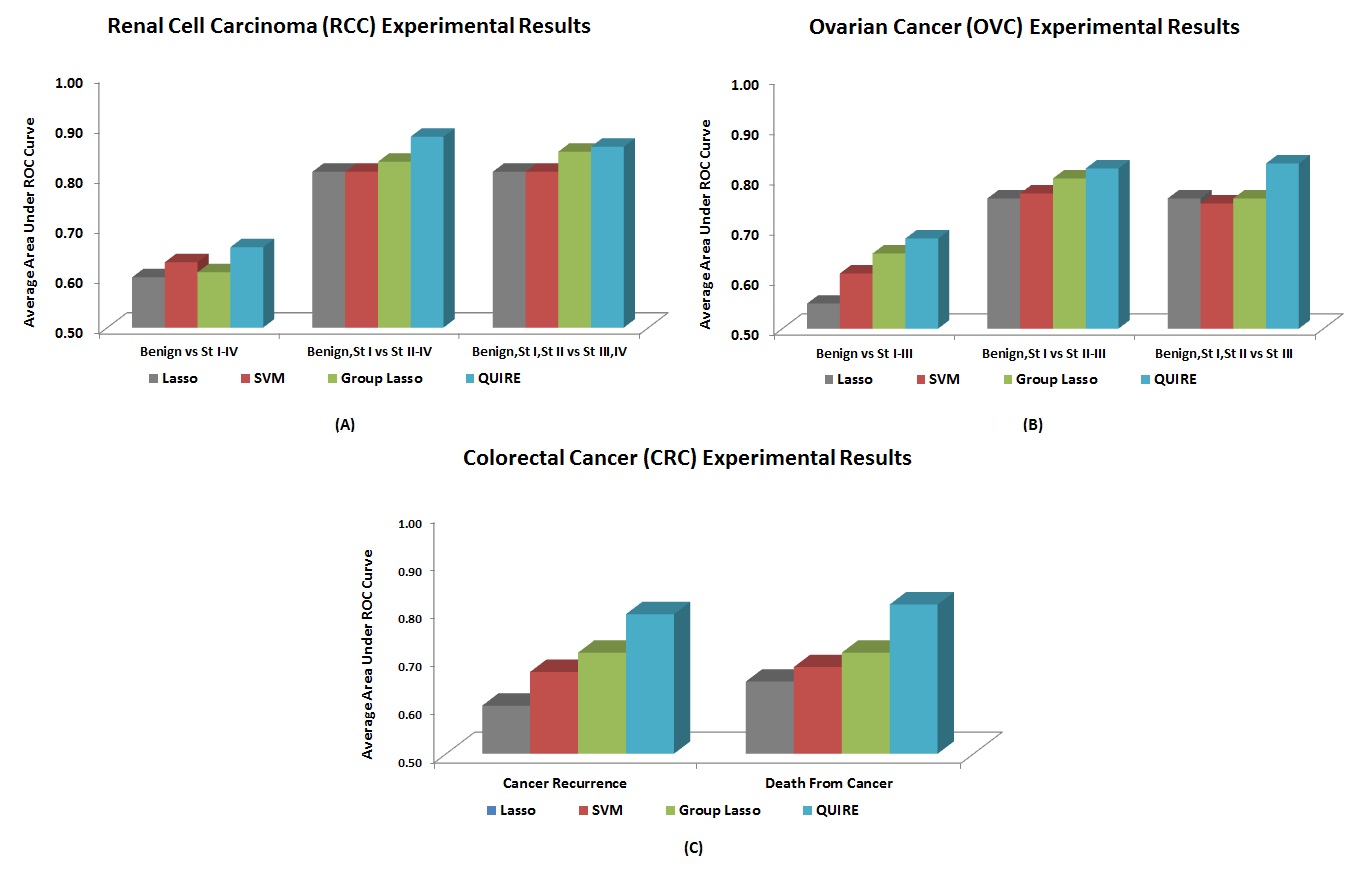}
\caption{Comparison of the classification performances of different
feature selection approaches with QUIRE in identifying the different stages of
(A)RCC , (B) OVC and (C) in predicting CRC recurrence and death from CRC.
In (A) and (C), five fold cross validation is repeated ten times and average
AUC score is reported. For (B), samples with \textit{CA125} marker's expression
level between $40$ and $100$ are used as test cases, while the remaining samples
are used for training. This experiment is also repeated ten times and average
AUC score is reported.
}
\label{fig:res}
\end{figure*}

In cancer stage prediction experiments for RCC and OVC,
we see from Figure~\ref{fig:res} that the combination of informative markers and
pairwise interactions identified by QUIRE show better classification performance
in every case, as compared to the
markers selected by Lasso, SVM and Overlapping Group Lasso.
For early detection of the diseases (classification of Benign, Stage $1$ vs. rest of the samples), QUIRE achieves average AUC scores of $0.88$ and $0.82$ for
RCC and OVC respectively. Overlapping group lasso shows next best
performance with average AUC scores of  $0.83$ and $0.80$ respectively.
Lasso and SVM, which do not use any grouping or interaction information amongst
the features, show the worst performance in all of the classification
tasks apart from one.
As QUIRE markers
show consistently better performance across all the stages of
RCC and OVC, they can be used for improved diagnosis and prognosis
of these two different types of cancers.
Also QUIRE helps better prediction of OVC progression for
samples with intermediate levels of \textit{CA125}; so it can be used by the physicians for early
detection of this disease.

From Figure~\ref{fig:res}(C), we can see that gene-gene interactions
help us better predict both CRC recurrence and death from CRC,
as compared to the other feature selection mechanisms.
In the events of cancer recurrence and death from cancer,
the average AUC values achieved by features selected with QUIRE
are $0.79$ and $0.81$ respectively, while markers identified by Overlapping Group Lasso
show the next best performance with average AUC value of $0.71$
in both of these categories. Markers identified by Lasso show the worst
performance in prediction of both of these events.
The performance
gap between QUIRE and the other three popular feature selection
techniques hint to the fact that QUIRE can identify interactions that
might help us better understand the mechanistic basis of CRC.

These experimental results show that QUIRE identifies markers and
interactions that complement each other in such a way that they not only
help better diagnosis and prognosis of cancer, but also can predict the
advanced events of recurrence of cancer and survival after cancer with higher
accuracy than other state-of-the-art algorithms. For each of these datasets,
identification of informative pairwise interactions using brute force enumerative technique is computationally impractical due to the huge dimensionality of the search space. QUIRE
helps reducing this space by a large margin. The total running time of QUIRE is
dominated by the Overlapping Group Lasso stage which takes around one hour to identify
biologically relevant groups of genes and protein interactions
in traditional desktop computers for the types of problems we study. After the
dimensionality is reduced, QUIRE exhaustively enumerates all the pairwise interactions and use
the protein interactions identified in the previous stage
on this low dimensional space in a couple of minutes.

\subsection{Informative QUIRE markers and interactions associated with cancer}
\vspace{-0.03 in}

Cancer is a genetic disease, which originates and develops through a process
of mutations.
Mutations in individual gene not only disrupts its own
function, but also
affects its interaction patterns with other genes.
As complex diseases like cancer is a result of dysregulation
in the interactions among the genes, researchers focus
on identifying those relevant interactions to gain more
insight into the molecular basis of the disease.
On the CRC dataset, QUIRE selects about $120$ quadratic
interactions on average as informative ones for both CRC recurrence
and death from CRC. On the other hand, the average number of markers selected
by Overlapping Group Lasso and Lasso on the same prediction tasks are about $1100$ and $150$ respectively.

An investigation of the pairwise interactions
identified by QUIRE on CRC dataset reveals that many of these interactions are indeed relevant
to the progression of cancer in general. Some of such
interactions identified for prediction of CRC recurrence include  \textit{JAK2 - LYN}\citep{samanta2009jak2}, Transforming growth
factor beta (\textit{TGF$\beta$}) - \textit{SMAD}\citep{grady2005transforming},
Epidermal growth factor receptor (\textit{EGFR}) - Caveolin (\textit{CAV})
\citep{dittmann2008radiation},
\textit{TP53} - TATA binding protein (\textit{TBP})\citep{crighton2003p53},
Connective tissue growth factor (\textit{CTGF}) - Vascular endothelial growth factor (\textit{VEGF})\citep{inoki2002connective}, Edoglin (\textit{ENG}) - Transforming growth
factor beta receptor (\textit{TGF$\beta$R})\citep{fonsatti2003endoglin}.
Further investigations of the interactions identified by QUIRE might  
reveal novel gene partners associated with cancer and thus lead
to testable hypothesis.

Disturbance in pairwise interactions among the genes affects the pathways in which they are located in.
Cancer pathways are a set of pathways dysregulations in which have been shown to be associated with
initiation and progression of the disease.  A pictorial view of the well-known cancer pathways can
be found in the KEGG database(\texttt{http://www.genome.jp/kegg/pathway/hsa\\/hsa05200.html})
\citep{kanehisa2012kegg}.
We perform a pathway enrichment analysis where we test if the
set of the markers and interactions identified by QUIRE on the CRC dataset reside in the cancer pathways.
As part of this experiment, we first use the partner genes identified by QUIRE
as part of the informative interactions while predicting CRC recurrence.
We use DAVID to identify the
statistically significant pathways that are enriched in these genes.
An investigation of the enriched
pathways returned by DAVID indicates that many of them are
indeed responsible for cancer or related to functions dysregulation
in which results in cancer. Some of such KEGG pathways include
Apoptosis(p-value $4.7$x$10^{-4}$), Focal adhesion(p-value $3$x$10^{-3}$),
Cell adhesion molecules(p-value $9.2$x$10^{-4}$),
p53 signaling pathway(p-value $1.3$x$10^{-2}$), Gap junction (p-value $1.3$x$10^{-2}$),
MAPK signaling pathway (p-value $4.5$x$10^{-2}$), ErbB signaling pathway (p-value $5.8$x$10^{-2}$), Cell cycle(p-value $6.6$x$10^{-2}$), Pathways in Cancer(p-value $7.2$x$10^{-4}$), Colorectal cancer(p-value $10^{-3}$).
Repeating the same analysis on the interacting partners identified
by QUIRE while predicting ``Death from CRC'' result in
identification of similar pathways (data not shown here).

Next we use the informative genes and their associated interactions discovered by QUIRE to
identify functional modules that might
be associated with pathways known to be dysregulated in cancer.
We use the web based tool Gene Mania (www.genemania.org)
\citep{warde2010genemania} to identify the statistically significant modules induced by
genes and interactions selected by QUIRE. Gene Mania also
returns the pathways and functions in which the identified
modules are significantly enriched.
After investigating these functional modules,
we find that many of them are enriched in
the well-known cancer pathways. Examples of such pathways include Focal adhesion
pathway (p-value $2$x$10^{-3}$), Jak-STAT signaling pathway (p-value $3$x$10^{-2}$), MAPK signaling pathway (p-value $1.4$x$10^{-3}$), NF-kappaB
signaling pathway (p-value $4.5$x$10^{-2}$), TGF beta signaling pathway (p-value $2.2$x$10^{-3}$) and Ras protein signaling pathway (p-value $1.3$x$10^{-2}$). Besides, some of the
induced modules are functionally enriched in processes disruptions in which
are known to be associated with initiation and progression of cancer. Some examples of such functions
include Apoptosis (p-value $4.2$x$10^{-3}$), Cell migration (p-value $1.3$x$10^{-3}$), Response
to growth factors (p-value $2.5$x$10^{-2}$), Cell cycle checkpoint (p-value $1$x$10^{-3}$), Cell-cell adhesion (p-value $3.1$x$10^{-3}$) etc.
We give examples of some of these modules in Figure~\ref{fig:subnets}.

\begin{figure*}[htp]
\centering
\includegraphics[width = 0.9\textwidth]{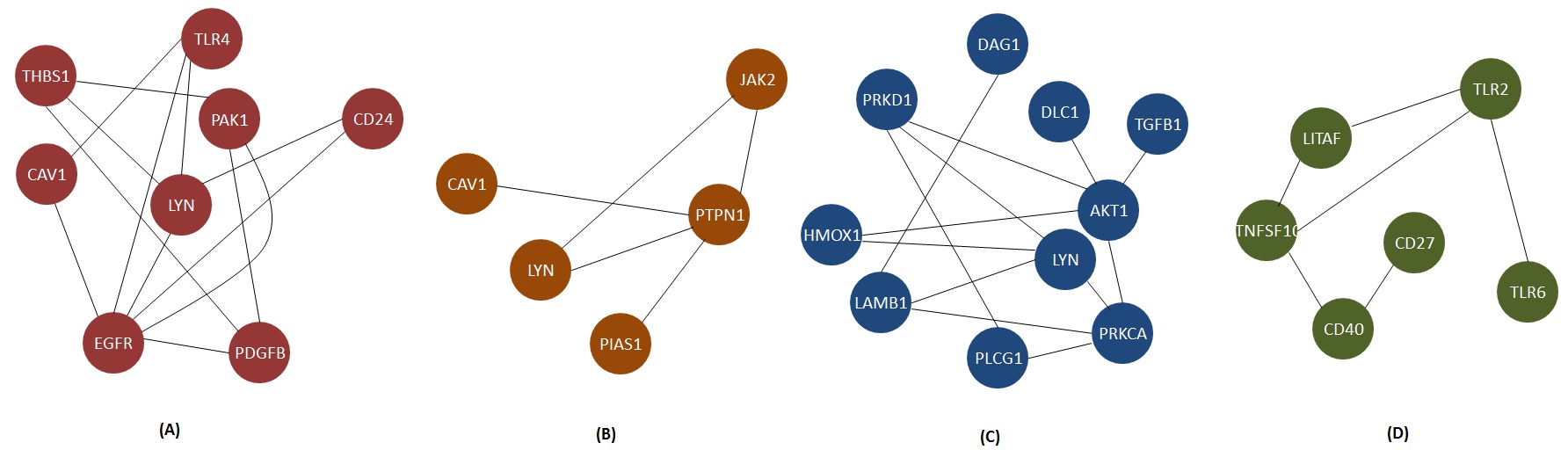}
\caption{Examples of functional modules, induced by markers and interactions
discovered by QUIRE and enriched in pathways and functions associated
with cancer. Module (A) reside in MAPK signaling pathway, module (B) in JAK-STAT
signaling pathway, module (D) in NF-kappaB signaling pathway and module (C)
is enriched in cell migration.
}
\label{fig:subnets}
\end{figure*}

\subsection{Characteristic Features of QUIRE}

\subsubsection{Effect of parameter $m$:} In this section, we perform
classification experiments to understand the effect of using different number
of input genes to QUIRE as identified by Overlapping Group Lasso. QUIRE
enumerates all possible pairwise interactions among these input genes
and returns the ones selected to most informative. For this experiment,
we increase the values of $m$ from $100$ to $400$ in steps of $100$
and we run QUIRE using these top $m$ genes. We use the whole set of protein protein interactions identified by Overlapping Group Lasso
in all the experiments. We report the average AUC score for ten runs of
five fold cross validation score on classifying the CRC samples for
prediction of CRC recurrence in Table \ref{tab:m_test}.
From this table, QUIRE shows best performance when the value of
$m$ is $200$. The performance gets worse as we increase the
$m$. The performance decrease may be attributed to the mechanism by which
Overlapping Group Lasso selects the groups. As sparsity is enforced
at the group level in Overlapping Group Lasso, some member genes
within the selected groups may be less informative about the phenotype and
as a result, their interactions result in
reduced prediction accuracy of QUIRE.

\begin{table}[h]
\vspace{-0.1 in}
\caption{Effect of parameter $m$ on classification performance of QUIRE}
\label{tab:m_test}
\begin{center}
\begin{tabular}{|c|c|}
\hline
Value of $m$ & Avg AUC score\\
\hline
    $100$ & $0.77$ \\
    $200$ & $0.79$ \\
    $300$ & $0.77$ \\
    $400$ & $0.73$ \\
\hline

\end{tabular}
\end{center}
\vspace{-0.1 in}
\end{table}

\subsubsection{Usage of different objective functions:}
We test the effect of using
different types of objective functions on the performance of QUIRE.
In our formulation of QUIRE, we use linear regression setting in the objective function.
The other option is to use
logistic regression. We perform classification experiments on CRC dataset which show that
the linear regression setting results in better performance on prediction of both CRC recurrence and
death from CRC, compared to the case where we use logistic regression formulation in the
objective function. We present the results from ten runs of five fold cross validation classification experiments in
Table \ref{tab:comp}, where we show average AUC score both for linear and logistic
regression formulation of QUIRE.

\begin{table}[h]
\vspace{-0.1 in}
\caption{Classification performance of QUIRE with Linear Regression and Logistic Regression formulations}
\label{tab:comp}
\begin{center}
\begin{tabular}{|c|c|c|}
\hline
Classification Experiment & Logistic Reg. & Linear Reg.\\
\hline
    CRC Recurrence & 0.69 & 0.79 \\
    Death from CRC & 0.69 & 0.81 \\
\hline

\end{tabular}
\end{center}
\vspace{-0.1 in}
\end{table}

\subsubsection{Quadratic feature interaction function:}
Next, we perform classification experiments to show the effect of
our proposed quadratic feature interaction  $g(x_j x_k)$ compared to $g(x_j) g(x_k)$.
The results from the ten runs of five fold cross validation classification experiments for prediction of  CRC recurrence and
death from CRC are shown in Table \ref{tab:int}. The AUC scores in the experimental
results show that our proposed original feature interaction formulation
$g(x_j x_k)$ is more effective in predicting post cancerous
events than the alternate weight sharing formulation
$g(x_j) g(x_k)$.

\begin{table}[h]
\vspace{-0.1 in}
\caption{Classification performance of QUIRE with different feature interaction mechanisms.}
\label{tab:int}
\begin{center}
\begin{tabular}{|c|c|c|}
\hline
Classification Experiment & $g(x_j) g(x_k)$ & $g(x_j x_k)$ \\
\hline
    CRC Recurrence & 0.71 & 0.79 \\
    Death from CRC & 0.70 & 0.81 \\
\hline

\end{tabular}
\end{center}
\vspace{-0.1 in}
\end{table}

\subsubsection{Performance of single gene features:}
We now illustrate the effects of using pairwise interaction in combination with the genes identified by Overlapping Group Lasso as informative ones. For this experiment, we use the group of genes discovered by Overlapping Group Lasso as ``Single Gene'' features in Lasso setting and compare it to the case where we additionally include pairwise interactions among these features and the protein protein interactions, as formulated in the objective function of QUIRE. We present results from ten runs of five fold cross validation classification experiments on CRC dataset in Table \ref{tab:sing}. The improved classification performance of the combination of single genes and pairwise interactions show that these interactions indeed provide additional information about the underlying biological phenomenon which cannot be captured by using single gene markers alone.

\begin{table}[h]
\vspace{-0.1 in}
\caption{Classification experiment to illustrate the effect of using pairwise gene interactions.}
\label{tab:sing}
\begin{center}
\begin{tabular}{|c|c|c|}
\hline
Classification Experiment & Single Gene Features & QUIRE \\
\hline
    CRC Recurrence & 0.70 & 0.79 \\
    Death from CRC & 0.70 & 0.81 \\
\hline

\end{tabular}
\end{center}
\vspace{-0.1 in}
\end{table}

\subsubsection{Performance of protein protein interactions:}
Finally, we perform classification experiments to observe the performance
of protein protein interactions on predicting CRC recurrence and death from
CRC. For this experiment, we use the single gene markers and the protein protein interactions selected
by Overlapping Group Lasso as input to QUIRE and
enumeration of the pairwise interactions among the marker genes
is avoided.
For ten runs of five fold cross validation experiment on this modified feature set,
we observe average AUC score of $0.71$ for both classification
tasks. These results show that besides physical interactions,
indirect higher level interactions among
the genes must be taken into account to understand
the basic mechanism of complex diseases.

\vspace{-0.2 in}
\section{Conclusion}
In this paper, we propose a computational approach, QUIRE, to identify combinatorial
interactions among the informative genes in complex diseases, like cancer. Our algorithm uses
Overlapping Group Lasso to identify functionally relevant gene markers and protein interactions
associated with cancer. It then explores the pairwise interactions among these relevant genes
within this reduced space exhaustively and the selected pairwise physical protein interactions to
discover the combination of individual markers and gene-gene interactions that are informative
for prediction of the disease status of interest. The application of QUIRE on three different types
of cancer samples collected using two different techniques shows that our approach performs
significantly better than the state-of-the-art feature selection methods such as Lasso and SVM
for biomarker discovery while selecting a smaller number of features, and it also shows that our
approach can capture discriminative interactions with high relevance to cancer progression.
Further investigations show that QUIRE can identify markers and interactions that have been
associated  previously with pathways associated with cancer. Moreover, the good performance
of QUIRE on the CRC dataset suggests that applications of QUIRE on genome-wide
microarray experimental data can be used to help prioritize Somamer design for
blood-based cancer diagnosis. And QUIRE applied to blood-based experimental data has the
great potential to impact the field of practical medical diagnosis.
 
However, QUIRE performs search for informative interactions on a lower dimensional space, which means that there is potential to miss interesting interactions relevant to the diseases. Extension of QUIRE to explore the whole genome wide interaction space will enable it to identify markers and interactions of more biological relevance. And more detailed investigations of the identified pairwise interactions among the informative genes will shed more light into the dynamics of the complex diseases.

\bibliographystyle{natbib}
\bibliography{recomb_refs}
\end{document}